\documentclass[letterpaper, 10 pt, journal, twoside]{IEEEtran}  

\IEEEoverridecommandlockouts                              

\usepackage{graphics} 
\usepackage{epsfig} 
\usepackage{mathptmx} 
\usepackage{times} 
\usepackage{amsmath} 
\usepackage{amssymb}  
\usepackage{graphicx}
\usepackage{cite}
\usepackage{caption}
\usepackage{subfigure}
\usepackage{epstopdf}
\usepackage{latexsym}
\usepackage[german,english]{babel}
\usepackage{todonotes}
\usepackage{soul}
\usepackage{color}
\graphicspath{{./images/}}

\title{
Recursive algorithm for the control of output remnant of Preisach hysteresis operator
}

\author{M. A. Vasquez-Beltran, B. Jayawardhana, R. Peletier
    \thanks{*This paper is based on research developed in the DSSC Doctoral Training Programme, co-funded through a Marie Skłodowska-Curie COFUND (DSSC 754315).} 
    \thanks{ $^{1}$M. A. Vasquez Beltran and B. Jayawardhana are with the Engineering and Technology Institute Groningen, Faculty of Science and Engineering, University of Groningen, 9747AG Groningen, The Netherlands {\tt\small \{m.a.vasquez.beltran;b.jayawardhana\} @rug.nl} }%
    \thanks{ $^{2}$R. Peletier is with the Kapteyn Astronomical Institute, Faculty of Science and Engineering, University of Groningen, 9747AG Groningen, The Netherlands {\tt\small r.peletier@rug.nl} }%
}

\newenvironment{proof}{\vspace{.1cm}\noindent{\sc Proof.   }\hspace{0.05cm}\,\,}{$\hfill\Box$\vspace{.1cm}}
\newtheorem{theorem}            {Theorem}[section]

\newtheorem{lemma}              [theorem]{Lemma} 
\newtheorem{proposition}		[theorem]{Proposition}

\newcommand{\dd}{{\rm d}\hbox{\hskip 0.5pt}}

\newcommand{\R}{{\mathbb R}}
\newcommand{\Z}{{\mathbb Z}}

\newcommand\intdl{\displaystyle\int\displaylimits}
\newcommand\iintdl{\displaystyle\iint\displaylimits}

\begin{document}

\pagestyle{empty}
\maketitle
\thispagestyle{empty}


\begin{abstract}

We study in this paper the control of hysteresis-based actuator systems where its remanence behavior (e.g., the remaining memory when the actuation signal is set to zero) must follow a desired reference point. 
We present a recursive algorithm for the output regulation of the hysteresis remnant behavior described by Preisach operators. Under some mild conditions, we prove that our proposed algorithm guarantees that the output remnant converges to a desired value. Simulation result shows the efficacy of our proposed algorithm.

\end{abstract}


\section{INTRODUCTION}
\IEEEPARstart{H}{ysteresis} is a complex non-linear behavior with particular memory characteristics and it is present in many physical systems such as shape memory alloys, mechanical systems with friction, and ferromagnetic and ferroelectric materials. Its influence becomes crucial and important when they are used in high-precision engineering systems. 

Hysteresis can occur as a quasi-static (rate-independent) or dynamic (rate-dependent) non-linear phenomenon which can mathematically be described by non-smooth integro-differential equations such as the Duhem hysteresis model \cite{Ikhouane2018}, infinite-dimensional operators such as the Preisach operator \cite{Preisach1935} or the combination thereof such as the Prandtl-Ishlinskii operator \cite{AlJanaideh2018}.  
Mathematical expositions of these hysteresis operators can be found, among many others, in  \cite{Brokate1996, Mayergoyz2003, Visintin1994, Hassani2014, Ismail2009}.

In the literature of systems and control theory, a number of methods have been proposed and studied to control nonlinear systems containing hysteretic sub-systems that can be described by one of the aforementioned hysteresis models. 
For instance, when the hysteretic element can be modeled by a classical (rate-independent) Preisach operator, a standard brute-force approach involves the identification and the use of inverse model that can approximately cancel the hysteresis non-linearity when it is connected in cascade \cite{Iyer2005}. 
An approach based on a multiplicative structure which does not require a direct inversion of a rate-dependent version of the Prandtl–Ishlinskii operator is presented in \cite{AlJanaideh2018}.
Other approaches exploit particular systems' properties and structure of the hysteresis model in order to design the stabilizing controller and to facilitate the analysis of the closed-loop systems. In this case, some well-studied systems' properties of hysteresis operators are dissipativity and passivity properties. 

In contrast to the aforementioned control problem where hysteresis is considered to be an undesirable nonlinear phenomenon, we study in this paper the control of the memory property of hysteresis operators. In particular, 
we are interested in the design of controller for regulating the output remnant value,
which is the leftover memory when the hysteresis input is set to zero, to a desired state.  
As hysteresis has a memory-effect that depends on the history of the applied input signal, the output remnant value can be driven from any given initial value to a desired one by a suitable input signal that is compactly defined (i.e., it has zero value outside a compact time interval). 
The set-point regulation of output remnant via a  compactly-defined input signal is relevant for applications that require minimal use of control input due to, for instance, input energy constraint or the associated energy loss/heat dissipation when a constant non-zero input is used to maintain the desired output. 

For high-precision mechatronic systems, a number of novel actuator systems have been proposed that exploit such output remnant behaviors. 
In \cite{Morita2007,Kadota2010}, 
a piezoelectric actuator with two stable configurations is developed.
A commercial piezoelectric actuator, so-called PIRest, is developed and presented in \cite{Reiser2018}. Recently, we have proposed and studied a hysteretic deformable mirror for space application that use a novel piezomaterial which allows us to achieve a large range of remnant deformation  \cite{Jayawardhana2018,Schmerbauch2020}. In the latter application, the use of set-point regulation via output remnant enables the development of a novel deformable mirror with high-density actuator systems via multiplexing with almost no heat dissipation \cite{Schmerbauch2020}.

In this paper, we propose a recursive algorithm to compute the desired compactly-defined input signal that solves the aforementioned set-point regulation problem using output remnant. We assume that the hysteresis is modeled by a classical Preisach operator and we use triangular signals as the basis for our compactly-defined input signal, similar to the one presented in \cite{Torres2018,Zhang2018}. Using our algorithm, we prove the asymptotic convergence of the signal to the desired one. Our results extend the work of \cite{Zhang2018} in two ways. Firstly, we show in Proposition 3.2 the existence of general sector bounds for the output remnant as a function of the amplitude of input signal without assuming sign-definiteness of the Preisach weighting function. Secondly, we show the monotonicity of the output remnant as a function of the amplitude within a compact interval such that the asymptotic convergence can be guaranteed in Proposition 4.2. Notably, the sign-indefiniteness of the Preisach's weighting function is relevant to the application of our algorithm to the output remnant control of piezoactuator systems that use piezomaterial exhibiting butterfly hysteresis loop as studied in  \cite{Jayawardhana2018}. 

\section{PRELIMINARIES}\label{sec:preliminaries}

We denote by $C(U,Y)$, $AC(U,Y)$, $C_{\text{pw}}(U,Y)$ the spaces of continuous, absolute continuous, and piece-wise continuous functions $f: U\to Y$, respectively.

\subsection{The Preisach hysteresis operator}
We introduce a formal definition of the classical Preisach operator following the exposition in \cite{Mayergoyz2003}. We define the so-called Preisach plane $P$ by $P:=\{(\alpha,\beta)\in\R^2\ |\ \alpha\geq\beta\}$, and correspondingly, we denote by $\mathcal{I}\subset P$ the set of all interfaces $L\in\mathcal{I}$, each of which is monotonically decreasing staircase line that can be described by a curve $\ell:\R_+\to P$ as follows $L=\{(\alpha,\beta)\ |\ (\alpha,\beta)=\ell(c),\, c\in\R_+ \}$ and such that $\ell(0)=(\beta_1,\beta_1)$ for some $\beta_1\in\R$, and $\lim_{c\to\infty} \Vert\ell(c)\Vert = \infty$. By monotonically decreasing we mean that $\alpha_1\geq\alpha_2$ whenever $\beta_1\geq\beta_2$ for all pairs $(\alpha_1,\beta_1),(\alpha_2,\beta_2)\in L$. Accordingly, the Preisach operator $\mathcal{P}:AC(\R_+,\R)\times\mathcal{I}\to AC(\R_+,\R)$ can be formally defined by
\begin{equation}\label{eq:operator_preisach}
    \Big( \mathcal{P} (u, L_0) \Big) (t) :=
    \iintdl_{(\alpha,\beta)\in P} 
    \mu(\alpha,\beta) 
    \left( \mathcal{R}_{\alpha,\beta} (u,L_0) \right) 
    (t)\ \dd\alpha \dd\beta
\end{equation}
where $\mu(\alpha,\beta)\in C(P,\R)$ is a weighting function, $L_0\in \mathcal{I}$ is the initial interface, and $\mathcal{R}_{\alpha,\beta}:AC(R_+,\R)\times\mathcal{I}\to C_{pw}(\R_+,\{-1,1\})$ is the relay operator defined by 
{\small\begin{equation}\label{eq:operator_relay}
    \Big(\mathcal R_{\alpha,\beta}(u,L_0)\Big)(t) :=
    \left\{ \begin{array}{ll}
        1 & \text{if } u(t)>\alpha, \\
        -1 & \text{if } u(t)<\beta, \\
        \big(\mathcal R_{\alpha,\beta}(u,r_0)\big)(t_-) 
            & \begin{aligned}
                \text{if } &\beta\leq u(t) \leq\alpha, \\ 
                &\quad\text{and } t>0,
                \end{aligned} \\
        r_{\alpha,\beta}(L_0) & \begin{aligned}
            \text{if } &\beta\leq u(t) \leq\alpha, \\
            &\quad\text{and } t=0.
            \end{aligned}
    \end{array}\right.
\end{equation}}
Note from the definition above that we have accommodated the initial interface $L_0$ through an auxiliary function $r_{\alpha,\beta}:\mathcal{I}\to\{-1,1\}$ which is defined by 
\begin{equation*}
    r_{\alpha,\beta}(L_0) := \left\{ \begin{array}{ll}
        1, & \text{if } L_0 \cap \{(\alpha_1,\beta_1)\,|\,\alpha<\alpha_1,\,\beta<\beta_1\} \neq\emptyset, \\
        -1 & \text{otherwise},
    \end{array}
    \right.
\end{equation*}
and whose purpose is to determine the initial state of the relay $\mathcal{R}_{\alpha,\beta}$ in accordance with the initial interface $L_0$. In other words, the function $r_{\alpha,\beta}$ will take value $+1$ if $(\alpha,\beta)$ is below the interface $L_0$, and $-1$ if $(\alpha,\beta)$ is above the interface $L_0$. It is important to note from \eqref{eq:operator_relay} that the value of $r_{\alpha,\beta}$ plays a role defining the initial state only for relays satisfying $\beta\leq u(0)\leq\alpha$.
Thus, to avoid inconsistencies between the value of $r_{\alpha,\beta}$ and the actual initial state some relays we assume always that $\left(u(0),u(0)\right)\in L_0$.

\subsection{The remnant control problem}

To introduce our formulation of the {\em remnant} control problem for the Preisach operator, let us start considering an input $u$ defined on a time interval $[0,\tau]$ with $\tau>0$ such that $u(0)=u(\tau)=0$, and an initial interface $L_0\in\mathcal{I}$ satisfying $(0,0)\in L_0$. When such input is applied to a Preisach operator in the form $\mathcal{P}(u,L_0)$, the final output value $y(\tau)$ may be different from the initial output value $y(0)$ due to the switching of some relays in the Preisach domain $P$ which occurs as result of the variations of $u$ within the interval $[0,\tau]$.
Let $L_\tau\in\mathcal{I}$ be the final interface which describes the state of relays in the Preisach operator at time instance $t=\tau$. It is clear that $(0,0)\in L_\tau$ (because $(u(\tau),u(\tau))=(0,0)$). Consequently, when the input of the Preisach operator is restricted to satisfy $u(0)=u(\tau)=0$, the initial and final interfaces are contained in a subset of $\mathcal{I}$ defined by
\begin{equation*}
    \mathcal{I}_\gamma := \left\{ L\in\mathcal{I}\,|\, (0,0)\in L \right\}.
\end{equation*}
Note that the restriction $u(0)=u(\tau)=0$ also compels relays whose $(\alpha,\beta)$ are in certain subdomains of $P$ to have fixed initial and final states regardless the behavior of $u$ within the interval $[0,\tau]$.
Consider a subdomain of the Preisach plane defined by
\begin{equation*}
    P_\gamma:=\{(\alpha,\beta)\in P\ |\ \alpha\geq 0,\, \beta\leq 0\}.
\end{equation*}
We have that every interface in $\mathcal{I}_\gamma$ lies entirely in $P_\gamma$. Consequently, relays whose $(\alpha,\beta)$ are not in the subdomain $P_\gamma$ are restrained to the state $-1$ (resp. $+1$) at both time instances $t=0$ and $t=\tau$ if they have $\beta>0$ (resp. $\alpha<0$). In other words, the set of relays  $\mathcal{R}_{\alpha,\beta}$ which have different initial and final state due to the variation of the signal $u$ in $(0,\tau)$ belongs to  $P_\gamma$.

The {\em remnant} of the Preisach operator refers to the instantaneous value of the output $y(t)$ when the input value satisfies $u(t)=0$ for some $t$. 
Roughly speaking, our remnant control problem corresponds to designing  a feedforward control input $u$ whose values at initial and terminal time are zero and the corresponding output of the Preisach operator has the desired remnant value $\gamma_d\in\R$ at the terminal time.   
To solve this problem, we propose a recursive algorithm based on an input of the form
\begin{equation}\label{eq:u_gamma}
    u_\gamma(t) := \sum_{k=0}^{\infty} w_k v_k(t)
\end{equation}
where $k\in\Z_+$, $w_k\in\R$ and $v_k$ is defined by
{\small\begin{equation}\label{eq:v}
    v_k(t) := \left\{\begin{array}{ccc}
        \frac{2}{\tau}(t-k\tau)  & \text{if }  k\tau \leq t \leq \left(k+\frac{1}{2}\right)\tau, \\[1mm]
        \frac{2}{\tau}\left(-t+(k+1)\tau\right) & \text{if } \left(k+\frac{1}{2}\right)\tau< t \leq\left(k+1\right)\tau, \\[1mm]
        0 & \text{otherwise},
    \end{array}\right.
\end{equation}}
with $\tau>0$. The function $v_k$ corresponds to a triangular pulse of unit amplitude and time length $\tau$, which starts at $t=k\tau$ and finishes at $t=(k+1)\tau$ and whose peak value occurs at $t=\left(k+\frac{1}{2}\right)\tau$. Therefore, the input $u_\gamma$ is a train of triangular pulses whose amplitudes are modulated by the factors $w_k$. 

Assume that $u_\gamma$ is applied as input to the Preisach operator and let $I_k\in\mathcal{I}_\gamma$ be the interface that describes the state of the relays at time instance $t=k\tau$ (i.e. $I_k = L_{(k\tau)}$). We can compute the remnant by a function $\gamma:\R\times\mathcal{I}_\gamma\to\R$ defined by
\begin{equation}\label{eq:gamma}
    \begin{aligned}
    \gamma(w_k,I_k) 
        &:= \Big(\mathcal{P} ( u_\gamma, I_0 ) \Big) \left((k+1)\tau\right) \\
        &\phantom{:}= \Big(\mathcal{P} ( w_k v_k, I_k) \Big) \left((k+1)\tau\right) = \Big(\mathcal{P} ( w_k v_0, I_k) \Big) \left(\tau\right)
    \end{aligned}
\end{equation}
In other words, the function $\gamma$ gives the remnant after the application of the $k$-th triangular pulse of $u_\gamma$ to a Preisach operator whose relays have initial states described by the interface $I_0$, or equivalently, the remnant after the application of a single triangular pulse with amplitude $w_k$ to a Preisach operator whose relays have initial states described by the interface $I_k$. In this way, we formulate the remnant control problem as finding the sequence of values $w_k$ that yields $\gamma(w_k,I_k)\to\gamma_d$ as $k\to\infty$.

\section{THE PROPERTIES OF THE REMNANT RATE}\label{sec:remnant_diff}

We analyze in this section the behavior of the remnant when the triangular pulses of the input $u_\gamma$ defined in \eqref{eq:u_gamma} is applied to the Preisach operator. For this, we consider the difference of remnant between two consecutive triangular pulses of $u_\gamma$, which is defined by
\begin{equation}\label{eq:Delta_k_gamma}
    \Delta_k \gamma := \gamma(w_{k+1},I_{k+1}) - \gamma(w_k,I_k).
\end{equation}
Let us introduce the auxiliary functions
{\small\begin{equation*}
    \begin{aligned}
        \ell_\beta^M(\alpha,L) &:= \max\left\{\beta\ |\ (\alpha,\beta)\in L\right\}, \\
        \ell_\beta^m(\alpha,L) &:= \min\left\{\beta\ |\ (\alpha,\beta)\in L\right\}, \\
        \ell_\alpha^M(\beta,L) &:= \max\left\{\alpha\ |\ (\alpha,\beta)\in L\right\}, \\
        \ell_\alpha^m(\beta,L) &:= \min\left\{\alpha\ |\ (\alpha,\beta)\in L\right\}, \\
    \end{aligned}
\end{equation*}}
which are used in following proposition to re-parameterize the coordinates $(\alpha,\beta)$ of the interface.

\begin{proposition}\label{prop:remnant_diff}
    Consider the remnant difference $\Delta_k \gamma$ defined in \eqref{eq:Delta_k_gamma}.
    For every $k\in\Z_+$, we have that
    {\small \begin{equation}\label{eq:Delta_k_gamma_explicit}
        \begin{aligned}
            \Delta_k \gamma = \left\{\begin{array}{cc}
                2\intdl_{M_{k+1}}^{w_{k+1}} \ 
                    \intdl_{\ell_\beta^M(\alpha,I_{k+1})}^{0} 
                    \mu(\alpha,\beta)\ \dd\beta\dd\alpha, & \text{if } w_{k+1}>M_{k+1},\\
                -2\intdl_{w_{k+1}}^{m_{k+1}} \ 
                    \intdl_{0}^{\ell_\alpha^m(\beta,I_{k+1})}
                    \mu(\alpha,\beta)\ \dd\alpha\dd\beta, & \text{if } w_{k+1}<m_{k+1},\\
                0, & \text{otherwise},
            \end{array}\right.
        \end{aligned}
    \end{equation}}
    with
    \begin{equation*}
        \begin{aligned}
            M_{k+1} = \ell_\alpha^M(0,I_{k+1})
            \quad\text{and}\quad
            m_{k+1} = \ell_\beta^m(0,I_{k+1}).
        \end{aligned}
    \end{equation*}
\end{proposition}
\begin{figure}
    \centering
    \subfigure[when $w_{k+1}>M_{k+1}$]{\includegraphics[width=0.40\linewidth]{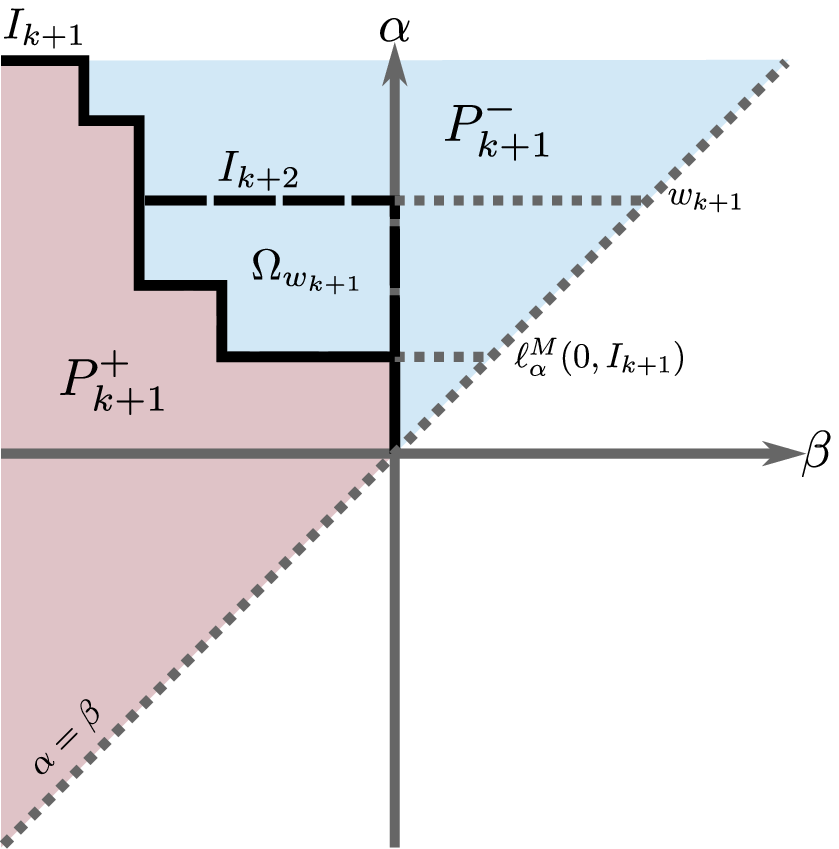}\label{sfig:pplane_w_k+1>M_k+1}}
    \hspace{5mm}
    \subfigure[when $w_{k+1}<m_{k+1}$]{\includegraphics[width=0.40\linewidth]{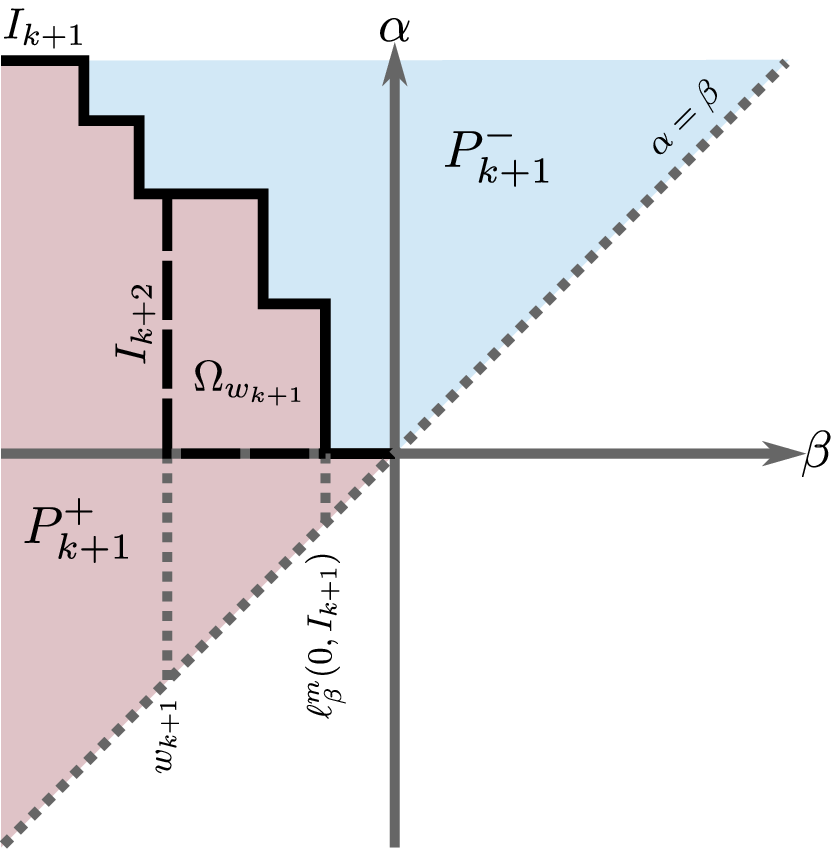}\label{sfig:pplane_w_k+1<m_k+1}}
    \caption{Partition of the Preisach plane $P$ used in Proposition \ref{prop:remnant_diff} to compute $\Delta_k \gamma := \gamma(w_{k+1},I_{k+1})-\gamma(w_k,I_k)$.\label{fig:pplane_w_k+1}}
\end{figure}%
\begin{proof}
    Consider the case when $w_{k+1}>M_{k+1}$ and let $P_{k+1}^+$ and $P_{k+1}^-$ be the subdomains of the Preisach domain $P$ that are below and above the interface $I_{k+1}$, respectively (see Fig. \ref{sfig:pplane_w_k+1>M_k+1}). Using these domains, the remnant of the Preisach operator at time instance $t=(k+1)\tau$ can be expressed by
    \begin{equation*}
        \begin{aligned}
            \gamma(w_k,I_k) 
            &= \iintdl_{P_{k+1}^+} \mu(\alpha,\beta)\ \dd\alpha\dd\beta - \iintdl_{P_{k+1}^-} \mu(\alpha,\beta)\ \dd\alpha\dd\beta. \\            
        \end{aligned}
    \end{equation*}
    Note that the value $M_{k+1}=\ell_\alpha^M(0,I_{k+1})$ is the $\alpha$-coordinate of the vertex in the interface $I_{k+1}$ which corresponds to last maximum of the input applied to the Preisach operator at time instance $t=(k+1)\tau$ (i.e. the last maximum of the truncated input $\{u_\gamma(t)\,|\,0\leq t\leq (k+1)\tau \}$). Therefore, since $w_{k+1}>M_{k+1}$, at the time instance $t=(k+2)\tau$ when the $(k+1)$-th triangular pulse finishes, there is a region $\Omega_{w_{k+1}}\subset P_{k+1}^-$ of relays whose states have switched from $-1$ to $+1$. This region is given by
    \begin{equation*}
        \Omega_{w_{k+1}} = \{(\alpha,\beta)\,|\,M_{k+1}\leq\alpha\leq w_{k+1},\ \ell_\beta^M(\alpha,I_{k+1})\leq\beta\leq 0\}.
    \end{equation*} 
    Consequently, it can be check that the remnant of the Preisach operator at time instance $t=(k+2)\tau$ is given by
    {\small\begin{equation*}
        \begin{aligned}
            \gamma(w_{k+1},I_{k+1}) 
            &=\phantom{+} \iintdl_{P_{k+1}^+} \mu(\alpha,\beta)\ \dd\alpha\dd\beta
            -\iintdl_{P_{k+1}^-}  \mu(\alpha,\beta)\ \dd\alpha\dd\beta \\
            &\phantom{=} + 2\iintdl_{\Omega_{w_{k+1}}} \mu(\alpha,\beta)\ \dd\alpha\dd\beta,
        \end{aligned}
    \end{equation*}}
    and subtracting both values of the remnant we have  
    \begin{equation*}
            \gamma(w_{k+1},I_{k+1})-\gamma(w_k,I_k) = 2\iintdl_{\Omega_{w_{k+1}}} \mu(\alpha,\beta)\ \dd\alpha\dd\beta,
    \end{equation*}
    and from the definition of the region $\Omega_{w_{k+1}}$, the integral limits can be parameterized as follows
    {\small\begin{equation*}
        \begin{aligned}
            \Delta_k \gamma = 2\intdl_{M_{k+1}}^{w_{k+1}} \ 
                \intdl_{\ell_\beta^M(\alpha,I_{k+1})}^{0} 
                \mu(\alpha,\beta)\ \dd\beta\dd\alpha .
        \end{aligned}
    \end{equation*}}
    Consider now the case when $w_{k+1}<m_{k+1}$ and again let $P_{k+1}^+$ and $P_{k+1}^-$ be the subdomains of the Preisach domain $P$ that are below and above the interface $I_{k+1}$, respectively (see Fig. \ref{sfig:pplane_w_k+1<m_k+1}). As in the previous case, the remnant of the Preisach operator at time instance $t=(k+1)\tau$ is given by
    \begin{equation*}
        \begin{aligned}
            \gamma(w_k,I_k) 
            &= \iintdl_{P_{k+1}^+} \mu(\alpha,\beta)\ \dd\alpha\dd\beta - \iintdl_{P_{k+1}^-} \mu(\alpha,\beta)\ \dd\alpha\dd\beta. \\            
        \end{aligned}
    \end{equation*}
    Observe that in this case the value $m_{k+1}=\ell_\beta^m(0,I_{k+1})$ is the $\beta$-coordinate of the vertex in the interface $I_{k+1}$ which corresponds to the last minimum of the input applied to the Preisach operator at time instance $t=(k+1)\tau$ (i.e. the last minimum of the truncated input $\{u_\gamma(t)\,|\,0\leq t\leq (k+1)\tau \}$). Since in this case $w_{k+1}<m_{k+1}$, at the time instance $t=(k+2)\tau$ when the $(k+1)$-th triangular pulse finishes, the region $\Omega_{w_{k+1}}\subset P_{k+1}^+$ of relays whose states have switched from $+1$ to $-1$ is given by
    \begin{equation*}
        \Omega_{w_{k+1}} = \{(\alpha,\beta)\,|\,0\leq\alpha\leq\ell_\alpha^m(\beta,I_{k+1}),\,w_{k+1}\leq\beta\leq m_{k+1}\},
    \end{equation*} 
    and the remnant of the Preisach operator at time instance $t=(k+2)\tau$ is given by
    {\small\begin{equation*}
        \begin{aligned}
            \gamma(w_{k+1},I_{k+1}) 
            &=\phantom{+} \iintdl_{P_{k+1}^+} \mu(\alpha,\beta)\ \dd\alpha\dd\beta
            -\iintdl_{P_{k+1}^-}  \mu(\alpha,\beta)\ \dd\alpha\dd\beta \\
            &\phantom{=} - 2\iintdl_{\Omega_{w_{k+1}}} \mu(\alpha,\beta)\ \dd\alpha\dd\beta .
        \end{aligned}
    \end{equation*}}
    Therefore, subtracting again both values of the remnant we have  
    \begin{equation*}
        \begin{aligned}
            \gamma(w_{k+1},I_{k+1})-&\gamma(w_k,I_k) = -2\iintdl_{\Omega_{w_{k+1}}} \mu(\alpha,\beta)\ \dd\alpha\dd\beta,
        \end{aligned}
    \end{equation*}
    and parameterizing the limits of the integral over the region $\Omega_{w_{k+1}}$ we have
    {\small\begin{equation*}
        \begin{aligned}
            \Delta_k \gamma = -2\intdl_{w_{k+1}}^{m_{k+1}} \ 
                \intdl_{0}^{\ell_\alpha^m(\beta,I_{k+1})}
                \mu(\alpha,\beta)\ \dd\alpha\dd\beta. \\
        \end{aligned}
    \end{equation*}}
    
    Finally, when $0\leq w_{k+1}<M_{k+1}$ or $m_{k+1}<w_{k+1}\leq 0$, then $I_{k+1}=I_{k+2}$ and at both time instances $t=(k+1)\tau$ and $t=(k+2)\tau$ all relays in the Preisach domain $P$ are in the same state which immediately implies $\gamma(w_{k+1},I_{k+1})-\gamma(w_k,I_k)=0$.
\end{proof}

Based on the explicit expression of $\Delta_k \gamma$ given by \eqref{eq:Delta_k_gamma_explicit} in Proposition \ref{prop:remnant_diff} and assuming that $\mu$ is compactly supported in a subset $P_\mu\subset P$, we can find sector bounds for $\Delta_k \gamma$ as a function of $\Delta_k w=w_{k+1}-w_k$. In other words, we find that the rate of the remnant difference respect to the difference between two consecutive amplitudes $w_{k+1}$ and $w_k$ is bounded.

\begin{proposition}\label{prop:remnant_diff_bounds}
    Let $\mu$ have a compact support $P_\mu\subset P$ whose intersection with $P_\gamma$ is not empty (i.e. $P_\mu\cap P_\gamma\neq\emptyset$), and consider $\Delta_k \gamma$ as given by \eqref{eq:Delta_k_gamma_explicit}. Then there exist constants $\Gamma_{1_+}\leq\Gamma_{2_+}$ and $\Gamma_{1_-}\leq\Gamma_{2_-}$ such that
    \begin{equation*}
        \begin{array}{cc}
            \Gamma_{1_+} \Delta_k w \leq\Delta_k \gamma \leq\Gamma_{2_+} \Delta_k w, & \text{if } \Delta_k w>0, \\[1mm]
            \Gamma_{1_-} \Delta_k w \leq\Delta_k \gamma  \leq\Gamma_{2_-} \Delta_k w, & \text{if } \Delta_k w<0,
        \end{array}
    \end{equation*}
    with $\Delta_k w = w_{k+1}-w_k$.
\end{proposition}
\begin{proof}
    Following analysis from Proposition \ref{prop:remnant_diff}, assume that $w_{k+1}>M_{k+1}=\ell_\alpha^M(0,I_{k+1})$. Then by taking the maximum and minimum of the inner integral in the first case of \eqref{eq:Delta_k_gamma_explicit}, we define
    {\small \begin{align}
        \Gamma_{1_+} &:= 2\min_{(\alpha,\beta_1)\in P_\mu\cap P_\gamma}\ \intdl_{\beta_1}^{0} \mu(\alpha,\beta)\ \dd\beta, \label{eq:Gamma_1_+} \\
        \Gamma_{2_+} &:= 2\max_{(\alpha,\beta_1)\in P_\mu\cap P_\gamma}\ \intdl_{\beta_1}^{0} \mu(\alpha,\beta)\ \dd\beta. \label{eq:Gamma_2_+}
    \end{align}}
    Note that since $\beta_1\leq 0$ for every $(\alpha,\beta_1)\in P_\gamma$, then either one of the values \eqref{eq:Gamma_2_+} or \eqref{eq:Gamma_1_+} is zero (i.e.  $\Gamma_{1_+}=0$ or $\Gamma_{2_+}=0$), or they have opposite signs (i.e.  $\Gamma_{1_+}<0<\Gamma_{2_+}$ ). Consequently, we find that
    {\small \begin{equation*}
        \begin{aligned}
            \Delta_k \gamma
            &\geq \intdl_{M_{k+1}}^{w_{k+1}} \Gamma_{1_+} \dd\alpha 
            = \Gamma_{1_+} (w_{k+1}-M_{k+1}), \\
            \Delta_k \gamma 
            &\leq \intdl_{M_{k+1}}^{w_{k+1}} \Gamma_{2_+} \dd\alpha
            = \Gamma_{2_+} (w_{k+1}-M_{k+1}).
        \end{aligned}
    \end{equation*}}
    Moreover, since $M_{k+1}=\ell_\alpha^M(0,I_{k+1})$ is the $\alpha$-coordinate of the vertex in the interface $I_{k+1}$ corresponding to the last maximum of the truncated input $\{u_\gamma(t)\ |\ 0\leq t\leq (k+1)\tau\}$, then we have that $w_k\leq M_{k+1}$, which leads us to
    \begin{equation*}
        \Gamma_{1_+} (w_{k+1}-w_k) \leq \Delta_k \gamma \leq \Gamma_{2_+} (w_{k+1}-w_k).
    \end{equation*}
    Analogously, for the case $w_{k+1}<m_{k+1}=\ell_\beta^m(0,I_{k+1})$, we take the maximum and minimum of the inner integral in the second case of \eqref{eq:Delta_k_gamma_explicit} and define
    {\small \begin{align}
        \Gamma_{1_-} &:= 2\max_{(\alpha_1,\beta)\in P_\mu\cap P_\gamma}\ \intdl_{0}^{\alpha_1} \mu(\alpha,\beta)\ \dd\alpha, \label{eq:Gamma_1_-}\\
        \Gamma_{2_-} &:= 2\min_{(\alpha_1,\beta)\in P_\mu\cap P_\gamma}\ \intdl_{0}^{\alpha_1} \mu(\alpha,\beta)\ \dd\alpha. \label{eq:Gamma_2_-}
    \end{align}}
    Similarly to the previous case, observe that since $\alpha_1\leq 0$ for every $(\alpha_1,\beta)\in P_\gamma$, then either one of the values \eqref{eq:Gamma_2_-} or \eqref{eq:Gamma_1_-} is zero (i.e.  $\Gamma_{1_-}=0$ or $\Gamma_{2_-}=0$), or they have opposite signs (i.e.  $\Gamma_{2_-}<0<\Gamma_{1_-}$). Therefore, in this case we have that
    {\small \begin{equation*}
        \begin{aligned}
            \Delta_k \gamma
            &\geq -\intdl_{w_{k+1}}^{m_{k+1}} \Gamma_{1_-} \dd\beta 
            = -\Gamma_{1_-} (m_{k+1}-w_{k+1}), \\
            \Delta_k \gamma 
            &\leq -\intdl_{w_{k+1}}^{m_{k+1}} \Gamma_{2_-} \dd\beta
            = -\Gamma_{2_-} (m_{k+1}-w_{k+1}). \\
        \end{aligned}
    \end{equation*}}
     Furthermore, in this case $m_{k+1}=\ell_\beta^m(0,I_{k+1})$ is the $\beta$-coordinate of the vertex in the interface $I_{k+1}$ corresponding to the last minimum of the truncated input $\{u_\gamma(t)\ |\ 0\leq t\leq (k+1)\tau\}$. Thus $w_k\geq m_{k+1}$ and we can obtain
    \begin{equation*}
        \Gamma_{1_-} (w_{k+1}-w_k) \leq \Delta_k \gamma \leq \Gamma_{2_-} (w_{k+1}-w_k).
    \end{equation*}
    Finally, when $m_{k+1}\leq w_{k+1}\leq M_{k+1}$ we have $\Delta_k \gamma = 0$ and both inequalities hold with the same values defined in \eqref{eq:Gamma_1_+}-\eqref{eq:Gamma_2_-}.
\end{proof}

Proposition \ref{prop:remnant_diff_bounds} proves the existence of general sector bounds for $\Delta_k \gamma$ as a function of $\Delta_k w$ disregarding the sign of $\mu$. In the next proposition we show that when $\mu$ is positive in a compact subset of $P_\gamma$, then under mild assumptions over the initial interface $I_0$ and the magnitude of every factor $w_k$, we have that $\Delta_k \gamma$ is monotonic respect to $\Delta_k w$.

\begin{proposition}\label{prop:remnant_diff_bounds_posivive_mu}
    Assume that there exists a non-empty subdomain $Q\subseteq P_\mu\cap P_\gamma$ of the form
    \begin{equation}\label{eq:Q}
        Q = \{(\alpha,\beta)\in P_\mu\cap P_\gamma\ |\ 0\leq\alpha\leq\alpha_2,\,\beta_2\leq\beta\leq 0 \},
    \end{equation}
    with $\alpha_2>0$ and $\beta_2<0$, such that $\mu(\alpha,\beta)\geq 0$ for every $(\alpha,\beta)\in Q$. Moreover, let the initial interface $I_0\in\mathcal{I}_\gamma$ be such that for every $(\alpha,\beta)\in I_0$ we have $\alpha\geq\alpha_2$ whenever $\beta\leq\beta_2$ and $\beta\leq\beta_2$ whenever $\alpha\geq\alpha_2$, and assume that $w_k\in[\beta_2,\alpha_2]$ for every $k\in\Z_+$. Then
    \begin{align}\label{eq:Delta_k_gamma_bound_positive_mu}
        0 & \leq \frac{\Delta_k\gamma}{\Delta_k w} \leq \max\left\{\Gamma_{2_+}^Q,\Gamma_{1_-}^Q\right\}, &\text{when } \Delta_k w \neq 0,
    \end{align}
    with
   {\small \begin{align}
        \Gamma_{2_+}^Q & = 2\max_{(\alpha,\beta_1)\in Q} \ \ \intdl_{\beta_1}^{0} \mu(\alpha,\beta)\ \dd\beta, \label{eq:Gamma_2_+^Q} \\ 
        \Gamma_{1_-}^Q & = 2\max_{(\alpha_1,\beta)\in Q} \ \ \intdl_{0}^{\alpha_1} \mu(\alpha,\beta)\ \dd\alpha. \label{eq:Gamma_1_-^Q}
    \end{align}}
\end{proposition}
\begin{proof}
    Note from the assumptions of the initial interface $I_0$ that none of its points lies in the subdomains $\{(\alpha,\beta)\ |\ \alpha>\alpha_2,\,\beta_2<\beta\leq 0\}$ and $\{(\alpha,\beta)\ |\ \beta<\beta_2,\,0\leq\alpha<\alpha_2\}$. Moreover, since $w_k$ is restricted to the interval $[\beta_2,\alpha_2]$ for every $k\in\Z_+$, then only the relays with $(\alpha,\beta)\in Q$ can be affected by the input $u_\gamma$ defined in \eqref{eq:u_gamma}. 
    Therefore, to find the sector bounds of $\Delta_k \gamma$ as a function of $\Delta_k w$, it is enough to modify \eqref{eq:Gamma_1_+}-\eqref{eq:Gamma_2_-} to take the maximum and minimum over $Q$. Thus when $\mu(\alpha,\beta)\geq 0$, for every $(\alpha,\beta)\in Q$,  we have that 
    {\small \begin{align*}
        \Gamma_{1_+}^Q &= 2\min_{(\alpha,\beta_1)\in Q}\ \intdl_{\beta_1}^{0} \mu(\alpha,\beta)\ \dd\beta = 0, \\
        \Gamma_{2_-}^Q &= 2\min_{(\alpha_1,\beta)\in Q}\ \intdl_{0}^{\alpha_1} \mu(\alpha,\beta)\ \dd\alpha = 0,
    \end{align*}}
    and it follows that
    \begin{align*}
        0 & \leq \Delta_k\gamma \leq \Gamma_{2_+}^Q\Delta_k w, &\text{if } \Delta_k w>0, \\
        \Gamma_{1_-}^Q\Delta_k w & \leq \Delta_k\gamma \leq 0, &\text{if } \Delta_k w<0,
    \end{align*}
    which combined yield \eqref{eq:Delta_k_gamma_bound_positive_mu}.
\end{proof}

We remark from Proposition \ref{prop:remnant_diff_bounds_posivive_mu} that in case the initial interface $L_0$ of a Preisach operator is unknown or does not satisfy the stated assumptions, it is possible to apply a single triangular pulse with amplitude either $w=\beta_2$ or $w=\alpha_2$ and to consider the new obtained interface, which will satisfy the assumptions, as the initial interface. Furthermore, when $\mu$ is negative in the set $Q$, an inequality to prove the monotonicity of $\Delta_k \gamma$ respect to $\Delta_k w$ can be also obtained. However, in that case we would obtain values $\Gamma^Q_{2_-}\leq 0$ and $\Gamma^Q_{1_+}\leq 0$ such that $\min\left\{\Gamma_{2_-}^Q,\Gamma_{1_+}^Q\right\} \leq \frac{\Delta_k\gamma}{\Delta_k w} \leq 0$. 

\section{THE RECURSIVE ALGORITHM FOR THE REMNANT CONTROL}\label{sec:algorithm}

In this section we present the recursive control algorithm to compute $w_{k+1}$ as a function of $w_{k}$ and the error of the remnant after the $k$-th triangular pulse of $u_\gamma$. Our algorithm works for the case considered in Proposition \ref{prop:remnant_diff_bounds_posivive_mu} when there exists a compact subset $Q\subset P_\mu\cap P_\gamma$ where $\mu$ is positive. The algorithm can easily be adapted to the case when $\mu$ is negative in a compact subset of $Q\subset P_\mu\cap P_\gamma$. Before introducing the algorithm, we present the next lemma which provides a way to compute the maximum and minimum remnant that can be obtained from a Preisach operator whose weighting function and initial interface satisfy conditions of Proposition \ref{prop:remnant_diff_bounds_posivive_mu}.

\begin{lemma}\label{lemma:max_remnant}
    Let $Q\subseteq P_\mu\cap P_\gamma$ and $I_0\in\mathcal{I}_\gamma$ be a non-empty subdomain and initial interface, respectively, that satisfy conditions stated in Proposition \ref{prop:remnant_diff_bounds_posivive_mu}. Then the maximum and minimum values of $\gamma$ with the initial interface $I_0$ are given by
    \begin{align}
        \gamma_{\max} = \max_{w\in[\beta_2,\alpha_2]} \gamma(w,I_0) = \gamma(\alpha_2,I_0), \label{eq:gamma_max}\\
        \gamma_{\min} = \min_{w\in[\beta_2,\alpha_2]} \gamma(w,I_0) = \gamma(\beta_2,I_0) \label{eq:gamma_min},
    \end{align}
    where $\alpha_2$ and $\beta_2$ are the values used for the definition of $Q$ in \eqref{eq:Q}.
\end{lemma}
\begin{proof}
    Note that since only relays with $(\alpha,\beta)\in Q$ can be affected by the input $u_\gamma$ when $w\in[\beta_2,\alpha_1]$, and $\mu$ is positive in $Q$, then the maximum (resp. minimum) remnant possible is obtained when all relays in $Q$ are in $+1$ state (resp. $-1$ state). It follows that after the application of a triangular pulse with amplitude $w=\alpha_2$ (resp. $w=\beta_2$), all relays in $Q$ are in $+1$ state (resp. $-1$ state).
\end{proof}
\begin{proposition}\label{prop:algorithm}
    Let $Q\subseteq P_\mu\cap P_\gamma$ and $I_0$ be a non-empty subdomain and initial interface, respectively, that satisfy conditions stated in Proposition \ref{prop:remnant_diff_bounds_posivive_mu}, and assume that $w_0 \in [\beta_2,\alpha_2]$ and $\gamma_d\in[\gamma_{\min},\gamma_{\max}]$. Consider the following update rule for the amplitude of the triangular pulse 
    \begin{equation}\label{eq:algorithm_w_k+1}
        w_{k+1} = w_k - \lambda e_k,
    \end{equation}
    where  $e_k=\gamma(w_k,I_k)-\gamma_d$ and $\lambda>0$ is the adaptation gain. If $\lambda$ satisfies
    {\small\begin{equation}\label{eq:algorithm_lambda_condition}
        0 < \lambda < \frac{2}{ \max\left\{\Gamma_{2_+}^Q,\Gamma_{1_-}^Q\right\} },
    \end{equation}}
    then $e_k\to0$ as $k\to\infty$.
\end{proposition}
\begin{proof}
    The remnant error after the application of the $(k+1)$-th triangular pulse is given by
    \begin{equation*}
        \begin{aligned}
            e_{k+1} 
            &= \gamma(w_{k+1},I_{k+1}) - \gamma_d \\
            &= \gamma(w_k,I_k) - \gamma_d + \gamma(w_{k+1},I_{k+1}) - \gamma(w_k,I_k)\\
            &= e_k + \Delta_k\gamma,
        \end{aligned}
    \end{equation*}
    where $\Delta_k \gamma$ is explicitly given by \eqref{eq:Delta_k_gamma_explicit} in Proposition \ref{prop:remnant_diff}. Introducing $\Delta_k w=-\lambda e_k$, we obtain
    {\small\begin{equation*}
        \begin{aligned}
            e_{k+1}
            = \left( e_k + \frac{\Delta_k\gamma}{\Delta_k w}\Delta_k w \right) 
            = \left( 1 - \lambda \frac{\Delta_k\gamma}{\Delta_k w} \right) e_k
        \end{aligned}
    \end{equation*}}
    which by Proposition \ref{prop:remnant_diff_bounds_posivive_mu} is a contraction mapping if $\lambda$ is chosen to satisfy \eqref{eq:algorithm_lambda_condition}.
\end{proof}

\section{SIMULATION}\label{sec:simulation}

To illustrate the application of the algorithm introduced in Proposition \ref{prop:algorithm}, we performed a simulation controlling the remnant of a particular class of Preisach operator known as the Preisach butterfly operator. The main characteristic of this class of Preisach operator is that its weighting function has disjoint subdomains of positive and negative values with a particular distribution and we refer interested readers to \cite{Jayawardhana2018} for the details. In this work, we used real data of the relation between electric-field and strain of a piezoelectric material 
sample made of doped Lead Zirconate Tinate (PZT) that exhibits the butterfly hysteresis loop on the left of Fig. \ref{fig:simulation_butterfly}.
The measurements were taken by laser interferometer applying triangular periodic inputs of 1400V of amplitude at constant low frequency of 1Hz, which is significantly lower than the resonant frequency of the system for obtaining the rate-independent hysteresis measurement as in \cite{Rakotondrabe2010_2}, and we fitted a weighting function to obtain the Preisach butterfly operator. 
For the obtained weighting function, the subdomain $Q$ was approximated by 
$Q = \{ (\alpha,\beta)\in P\ | \ -850\leq\beta\leq 0,\, 0\leq\alpha\leq 1400\},$
which is indicated by a dashed line enclosing a region of the weighting function illustrated in Fig. \ref{fig:simulation_butterfly}. We found for this $Q$ that $\Gamma_{2_+}\approx 6.83$, $\Gamma_{1_-}\approx 5.50$, $\gamma_{\max}\approx 433.83$, and $\gamma_{\min}=\approx -141.96$, and the initial interface considered was $I_0=\{(\alpha,\beta)\in P\ |\ \alpha=1400,\,-\infty<\beta\leq-800\} \cup \{(\alpha,\beta)\in P\ |\ 0\leq\alpha\leq 1400,\, \beta=-800\}$. 
For simulation purpose, we took $\lambda = 0.28$ and $\gamma_d=250$ 
and used an input $u_\gamma$ whose triangular pulses length was $\tau=1$. We truncated it to zero after $20$ steps (i.e. $u(t)=0$ for $t\geq 20$) once the output remnant $\gamma(w_k,I_k)$ was sufficiently close to $\gamma_d$. 
It can be observed in the simulation results of Fig. \ref{fig:simulation_results} that the output value $y(t)\approx\gamma_d$ is maintained for $t\geq 20$ when the input $u_\gamma$ has been removed.
\begin{figure}
    \centering
    \subfigure{\includegraphics[width=0.22\textwidth,trim={0 0.75cm 0 0.75cm},clip]{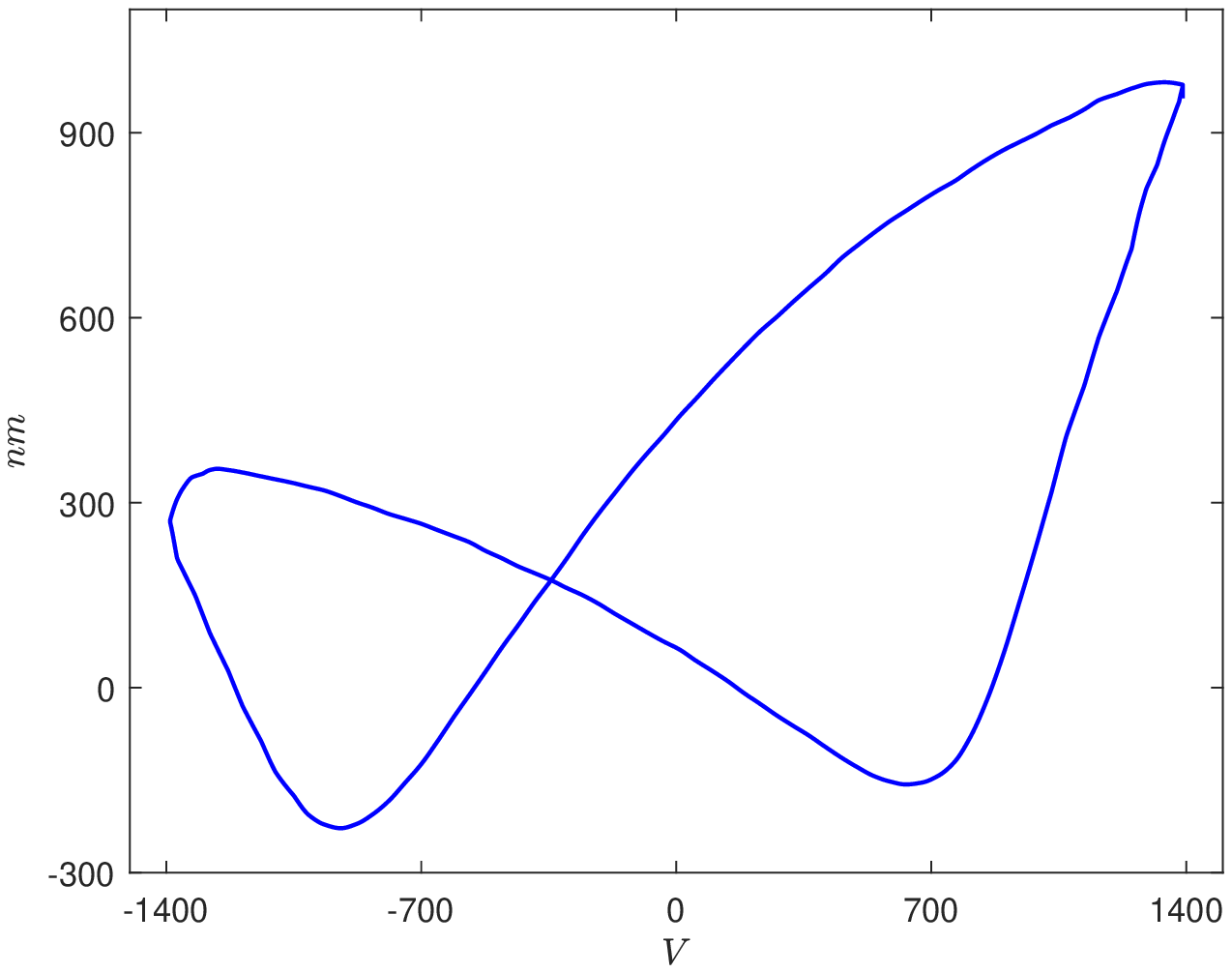}}
    \subfigure{\includegraphics[width=0.22\textwidth,trim={0 0.75cm 0 0.75cm},clip]{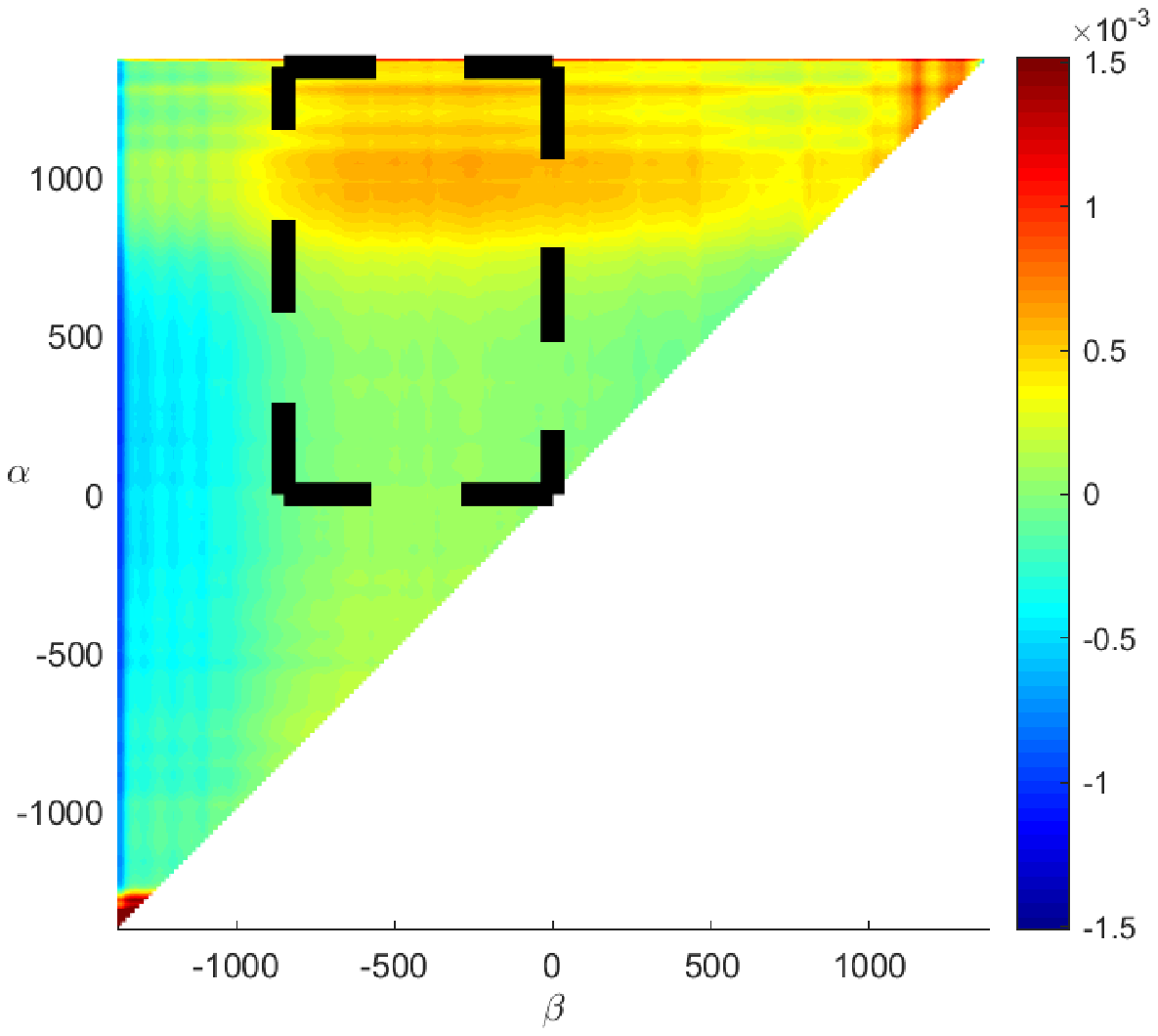}}
    \caption{Experimental butterfly hysteresis loop exhibited in the relation between voltage (V) and strain (nm) of a piezoelectric material and the corresponding weighting function $\mu$ of the fitted Preisach butterfly operator with the region $Q$ where $\mu$ is positive enclosed by the dashed line.}
    \label{fig:simulation_butterfly}
\end{figure}
\begin{figure}
    \centering
    \includegraphics[width=0.90\linewidth,trim={0 0.75cm 0 0.75cm},clip]{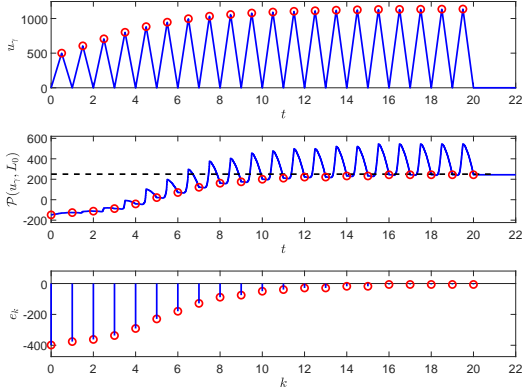}
    \caption{Simulation results for the first 20 steps of the algorithm controlling the remnant of a Preisach butterfly operator with an input $u_\gamma$ whose triangular pulses length is $\tau=1$. The upper plot shows the input $u_\gamma(t)$ where the amplitude $w_k$ of the $k$-th triangular pulse is marked in red. The middle plot corresponds to output $y(t)$ with the remnant $\gamma(w_k,I_k)$ marked in red. The bottom plot shows the remnant error $e_k=\gamma(w_k,I_k)-\gamma_d$.}
    \label{fig:simulation_results}
\end{figure}%

\section{CONCLUSIONS}\label{sec:conclusions}
In this work we presented a formulation for the problem of controlling the {\em remnant} of a system with hysteresis modeled by a Preisach operator. Using 
train of triangular pulses as the kernel of the remnant control input $u$, we analyze the properties of output remnant sequences due to the application of this family of input signals to the Preisach operator. Subsequently, we present recursive algorithm to update the amplitude of the triangular pulse sequences that guarantees the convergence of the output remnant sequence to a desired remnant value under some mild conditions.

\addtolength{\textheight}{-3cm}   

\section{ACKNOWLEDGMENTS}

We would like to thank prof. M. Acuautla and prof. B. Noheda of the University of Groningen for providing us with the experimental data of the piezoelectric material. 



\bibliographystyle{IEEEtran}
\bibliography{library}

\end{document}